\documentclass{IEEEtran}

\usepackage{cite}
\usepackage{amsmath,amssymb,amsfonts,amsthm,newtxmath,mleftright,siunitx}
\usepackage{graphicx,xcolor,booktabs}

\renewcommand{\vec}[1]{\ensuremath{\mathbf{#1}}} 
\newcommand{\mat}[1]{\ensuremath{\mathbf{#1}}} 
\newcommand{\T}{\ensuremath{\mathsf{T}}} 

\newcommand{\erf}{\operatorname{erf}}
\newcommand{\erfc}{\operatorname{erfc}}

\newcommand{\pd}{\ensuremath{p_\mathit{d}}}
\newcommand{\pfa}{\ensuremath{p_\mathit{fa}}}

\newtheorem{proposition}{Proposition}

\theoremstyle{definition}

\theoremstyle{remark}
\newtheorem*{remark}{Remark}

\begin{document}
	
\title{Noise-Type Radars: Probability of Detection vs.\ Correlation Coefficient and Integration Time}

\author{
	\IEEEauthorblockN{David Luong,~\IEEEmembership{Graduate Student Member, IEEE}, Bhashyam Balaji,~\IEEEmembership{Senior Member, IEEE}, and \\ Sreeraman Rajan,~\IEEEmembership{Senior Member, IEEE}}
	
	\thanks{D.\ Luong is with Carleton University, Ottawa, ON, Canada K1S 5B6. Email: david.luong3@carleton.ca.}
	\thanks{B.\ Balaji is with Defence Research and Development Canada, Ottawa, ON, Canada K2K 2Y7. Email: bhashyam.balaji@drdc-rddc.gc.ca.}%
	\thanks{S.\ Rajan is with Carleton University, Ottawa, ON, Canada K1S 5B6. Email: sreeraman.rajan@carleton.ca.}%
}

\maketitle

\begin{abstract}
	Noise radars have the same mathematical description as a type of quantum radar known as \emph{quantum two-mode squeezing radar}. Although their physical implementations are very different, this mathematical similarity allows us to analyze them collectively. We may consider the two types of radars as forming a single class of radars, called \emph{noise-type radars}. The target detection performance of noise-type radars depends on two parameters: the number of integrated samples and a correlation coefficient. In this paper, we show that when the number of integrated samples is large and the correlation coefficient is low, the detection performance becomes a function of a single parameter: the number of integrated samples multiplied by the square of the correlation coefficient. We then explore the detection performance of noise-type radars in terms of this emergent parameter; in particular, we determine the probability of detection as a function of this parameter.
\end{abstract}


\section{Introduction}

Quantum radar, which has the potential to improve sensing performance compared to conventional radars, has been studied from a theoretical point of view for more than a decade \cite{lloyd2008qi,tan2008quantum}. But it was only in the last several years that quantum radar has begun making the transition from theory to experiment. After several quantum lidar experiments \cite{lopaeva2013qi,zhang2015entanglement,balaji2018qi}, this push toward the experimental realm culminated in an experiment by Wilson et al.\ that validated the concept of a practically-realizable quantum radar design at microwave frequencies \cite{chang2018quantum,luong2019roc}. Shortly after, the results of the Wilson experiment were confirmed by a replication experiment by a group based in the Institute of Science and Technology Austria \cite{barzanjeh2019experimental}. The radar architecture described in \cite{luong2019roc} is called \emph{quantum two-mode squeezing radar} (QTMS radar).

One of the most interesting properties of QTMS radar is that it forms a link between quantum radar and standard noise radars of the type previously studied in the literature \cite{cooper1967random,lukin1998millimeter,narayanan2004design,thayaparan2006noise,lukin2008Ka,tarchi2010SAR,narayanan2016noise,savci2019trials,savci2020noise}. Although the physical implementation of noise radars and QTMS radars are very different, they have similar mathematical descriptions \cite{luong2019cov}. They are both characterized by a correlation coefficient $\rho$, with QTMS radars generally achieving higher values of $\rho$ than standard noise radars. The mathematical connection between the two classes of radars means that QTMS radars inherit some of the desirable properties of noise radars, such as low-probability-of-intercept operation and excellence in spectrum sharing \cite{kulpa2013signal,wasserzier2019noise,luong2020magazine}. Thus, we may consider QTMS radars and noise radars as subsets of a larger class of radars. We will call this class of radars \emph{noise-type radars}, because both classes of radars employ noise as a transmit signal.

In \cite{luong2020int}, we found, as a general rule, that increasing $\rho$ by a factor of $\alpha$ is equivalent to increasing the number of integrated samples $N$ by a factor of $\alpha^2$. This heuristic is particularly easy to use in practice because $\rho$ can be estimated easily from a single set of radar detection data \cite{luong2022structured}; unlike the signal-to-noise ratio, there is no need to measure the signal and the noise separately. It explains the experimental observation in \cite{luong2019roc} that the detection performance QTMS radar prototype was equivalent to that of a comparable classical noise radar whose integration time was increased by a factor of 8. In the present paper, we extend the analysis in \cite{luong2020int} by showing that the detection performance of noise-type radars depends on the combination $N\rho^2$ for a variety of detectors; in \cite{luong2020int} we only considered one detector. We also show how $N\rho^2$ can be used as an aid to designing noise-type radars with given performance requirements.

The remainder of the paper is organized as follows. After introducing some of the basic theory of noise-type radars in Sec.\ \ref{sec:background}, we introduce three detectors suitable for noise radar target detection in Sec.\ \ref{sec:detectors}. In Sec.\ \ref{sec:detection_perf}, we show that the receiver operating characteristic (ROC) curves for these detectors depend on $N$ and $\rho$ only through the combination $N\rho^2$ when $N$ is large and $\rho$ is small. In Sec.\ \ref{sec:pd}, we show plots of the probability of detection ($\pd$) as a function of $N\rho^2$. These plots were inspired by the commonly-used plots of $\pd$ vs.\ SNR, of which \cite[Fig.\ 2.6]{skolnik1962introduction} is an example. Both types of plots serve the role of showing, at a glance, what parameter values are required to achieve a given detection probability. We then give approximations of $\pd$ that do not rely on special functions like the Marcum $Q$-function. We also explore what happens in the small-$N$ limit, where the assumptions used to derive the results in this paper break down. Finally, Sec.\ \ref{sec:conclusion} concludes the paper.

\section{Background}
\label{sec:background}

The noise-type radars we consider in this paper have two electromagnetic signals associated with them: the \emph{received signal} and a \emph{reference signal} used for matched filtering. This reference signal is, in principle, a copy of the transmit signal. However, we consider the reference signal instead of the transmit signal because only the reference signal is available to the radar receiver, and it may not be an absolutely perfect copy of the transmit signal.

In the following, we assume that the transmit and reference signals are white noise processes and that they are jointly Gaussian. We further assume that any external noise (e.g.\ system noise or atmospheric noise) can be modeled as additive white Gaussian noise. Under these assumptions, the received and reference signals are likewise white noise processes that are jointly Gaussian. We denote the in-phase and quadrature voltages of the received signal by $I_1[n]$ and $Q_1[n]$, respectively, and the corresponding voltages of the reference signal by $I_2[n]$ and $Q_2[n]$. Because these four time series, being white noise, are uncorrelated in time, we may drop the time index $n$ and consider only the case where the time lag is zero. 

These considerations allow us to model the signals associated with a noise-type radar by the real-valued random vector $\vec{x} = [I_1, Q_1, I_2, Q_2]^\T$, where $\vec{x}$ follows a multivariate Gaussian distribution with zero mean and covariance matrix
\begin{equation} \label{eq:NR_cov}
	\mat{\Sigma}(\sigma_1, \sigma_2, \rho, \phi) =
	\begin{bmatrix}
		\sigma_1^2 \mat{1}_2 & \rho \sigma_1 \sigma_2 \mat{R}(\phi) \\
		\rho \sigma_1 \sigma_2 \mat{R}(\phi)^\T & \sigma_2^2 \mat{1}_2
	\end{bmatrix} \! .
\end{equation}
This matrix, which was derived in \cite{luong2019cov}, depends on the following parameters:
\begin{itemize}\setlength{\itemsep}{2pt}
	\item $\sigma_1^2$: the power of the received signal
	\item $\sigma_2^2$: the power of the reference signal
	\item $\rho$: a coefficient characterizing the correlation between the received and reference signals
	\item $\phi$: the phase between the two signals.
\end{itemize}
Furthermore, $\mat{1}_2$ is the $2 \times 2$ identity matrix and $\mat{R}(\phi)$ is the rotation matrix 
\begin{equation}
	\mat{R}(\phi) = 
	\begin{bmatrix}
		\cos \phi & \sin \phi \\
		-\sin \phi & \cos \phi
	\end{bmatrix} \! .
\end{equation}
QTMS radars are described by a matrix of the same overall form, but with the reflection matrix
\begin{equation}
	\mat{R}'(\phi) = 
	\begin{bmatrix}
		\cos \phi & \sin \phi \\
		\sin \phi & -\cos \phi
	\end{bmatrix} \! .
\end{equation}
substituted for $\mat{R}(\phi)$. The results in this paper do not depend on whether $\mat{R}(\phi)$ or $\mat{R}'(\phi)$ is used.

\section{Detectors for Target Detection}
\label{sec:detectors}

Out of the four parameters that appear in the covariance matrix \eqref{eq:NR_cov}, the correlation coefficient $\rho$ is the most important for target detection. This is because, in the absence of clutter, $\rho$ is nonzero only when there a target is present to reflect some of the transmitted signal to the radar receiver. In the absence of a target, the received signal would consist only of external noise that is uncorrelated with the reference signal. Hence, the target detection problem for noise-type radars is equivalent to deciding between the following hypotheses:
\begin{equation} \label{eq:hypotheses}
	\begin{alignedat}{3}
		H_0&: \rho = 0 &&\quad\text{Target absent} \\
		H_1&: \rho > 0 &&\quad\text{Target present.}
	\end{alignedat}
\end{equation}

In previous work, several test statistics---or \emph{detectors}---for the noise radar target detection problem were studied. We will consider three of them here. Before we do so, however, we establish the following notation. For any random variable $X$, we denote its sample mean (calculated from $N$ independent samples $x_1, \dots, x_n$ drawn from the distribution of $X$) by an overline:
\begin{equation}
	\bar{X} = \frac{1}{N} \sum_{n = 1}^N x_n.
\end{equation}
The variable $N$ will always be reserved for the number of samples used to calculate the sample mean. More specifically, since we are working in the context of radars, we say that $N$ is the number of samples integrated by the radar. 

We also define the following auxiliary quantities:
\begin{subequations}
	\begin{align}
		\label{eq:aux_P1}
		P_1 &= I_1^2 + Q_1^2 \\
		\label{eq:aux_P2}
		P_2 &= I_2^2 + Q_2^2 \\
		\label{eq:aux_Rc}
		R_c &= I_1 I_2 \pm Q_1 Q_2 \\
		\label{eq:aux_Rs}
		R_s &= I_1 Q_2 \mp I_2 Q_1.
	\end{align}
\end{subequations}
The upper signs apply when the rotation matrix $\mat{R}(\phi)$ is used in \eqref{eq:NR_cov}; the lower signs apply when the reflection matrix $\mat{R}'(\phi)$ is used. 

We are now in a position to state the three detectors which will be considered in this paper.

\begin{proposition} \label{prop:detNP}
	If $\sigma_1 = \sigma_2 = 1$ and $\phi = 0$, the locally most powerful test statistic for $\rho$ close to zero is
	\begin{equation}
		D_0 = \bar{R}_c.
	\end{equation}
\end{proposition}

\begin{proof}
	See \cite[Sec.\ V]{luong2022family}.
\end{proof}

\begin{proposition}
	The maximum likelihood estimator for $\rho$, which is also the generalized likelihood ratio test statistic for the hypotheses \eqref{eq:hypotheses}, is
	\begin{equation}
		\hat{\rho} = \sqrt{ \frac{\bar{R}_c^2 + \bar{R}_s^2}{\bar{P}_1 \bar{P}_2} }.
	\end{equation}
\end{proposition}

\begin{proof}
	See \cite[Sec.\ V]{luong2022structured}.
\end{proof}

\begin{proposition}
	The detector that arises from treating $I_1[n] + j Q_1[n]$ as a complex-valued time series and performing matched filtering with the reference signal $I_2[n] \pm j Q_2[n]$ is
	\begin{equation}
		D_\mathrm{MF} = \sqrt{\bar{R}_c^2 + \bar{R}_s^2}.
	\end{equation}
	As above, the sign to be chosen depends on whether $\mat{R}(\phi)$ or $\mat{R}'(\phi)$ is used in \eqref{eq:NR_cov}.
\end{proposition}

\begin{proof}
	Because we need only consider the case where the time lag between the two signals is zero, determining the output of the matched filter reduces to calculating
	\begin{equation}
		\overline{(I_1 + j Q_1)(I_2 \pm j Q_2)^*} = \bar{R}_c \mp j \bar{R}_s.
	\end{equation}
	For target detection, we need only consider the magnitude of the matched filter output, resulting in $D_\mathrm{MF}$.
\end{proof}

\section{Detection Performance and $N\rho^2$}
\label{sec:detection_perf}

In this section, we prove that all of the detectors listed in the previous section share one commonality: when $N$ is large and $\rho$ is small, the two parameters $N$ and $\rho$ ``coalesce'', and their ROC curves depend only on the single parameter $N \rho^2$. Note that this behavior is not limited to these three detectors alone; we have chosen them only as representative examples.

\begin{proposition} \label{prop:rhoN2_detNP}
	In the limit $N \to \infty$ and to first order in $\rho$, the ROC curve for $D_0$ depends on $N$ and $\rho$ only through $N\rho^2$.
\end{proposition}

\begin{proof}
	In the limit of large $N$, we showed in \cite{luong2022family} that the ROC curve is approximately
	\begin{equation}
		\pd(\pfa) = \frac{1}{2} \erfc \mleft( \frac{ \erfc^{-1}(2 p_\mathit{fa}) - \sqrt{N} \rho \cos \phi }{ \sqrt{1 + \rho^2 \cos^2 \phi} } \mright).
	\end{equation}
	where $\pd$ is the probability of detection, $\pfa$ is the probability of false alarm, and $\erfc(\cdot)$ is the complementary error function. Keeping only first-order terms in $\rho$, this reduces to
	\begin{equation} \label{eq:ROC_0_small}
		p_\mathit{d}(p_\mathit{fa}) \approx \frac{1}{2} \erfc \! \big[ \! \erfc^{-1}(2 p_\mathit{fa}) - \sqrt{N} \rho \cos \phi \big],
	\end{equation}
	which depends on $N$ and $\rho$ only through $N\rho^2$.
\end{proof}

\begin{remark}
	It is clear from \eqref{eq:ROC_0_small} that $D_0$ is a phase-dependent detector. However, as stated in Prop.\ \ref{prop:detNP}, this detector is locally most powerful only when $\phi = 0$. Therefore, for the remainder of this paper, we will take $\phi = 0$ whenever we consider $D_0$.
\end{remark}

\begin{proposition} \label{prop:rhoN2_detRho}
	In the limit $N \to \infty$ and to first order in $\rho$, the ROC curve for the $\hat{\rho}$ detector depends on $N$ and $\rho$ only through $N\rho^2$.
\end{proposition}

\begin{proof}
	We showed in \cite{luong2019rice} that, when $N$ is greater than approximately 100, the ROC curve for $\hat{\rho}$ is
	\begin{equation} \label{eq:ROC_detRho}
		\pd(\pfa | \rho, N) = Q_1 \mleft( \frac{\rho \sqrt{2N}}{1 - \rho^2}, \frac{\sqrt{-2 \ln \pfa}}{1 - \rho^2} \mright)
	\end{equation}
	where $Q_1(\cdot, \cdot)$ is the Marcum $Q$-function of order 1 (not to be confused with the quadrature voltage $Q_1$); see also \cite{luong2022structured}. To first order in $\rho$, we have
	\begin{equation} \label{eq:ROC_detRho_small}
		\pd(\pfa | \rho, N) \approx Q_1 \mleft( \rho \sqrt{2N}, \sqrt{-2 \ln \smash{\pfa}} \mright).
	\end{equation}
	This expression depends on $N$ and $\rho$ only through $N\rho^2$, so the proposition holds.
\end{proof}

\begin{proposition} \label{prop:rhoN2_detMF}
	In the limit $N \to \infty$ and to first order in $\rho$, the ROC curve for $D_\mathrm{MF}$ depends on $N$ and $\rho$ only through $N\rho^2$.
\end{proposition}

\begin{proof}
	Under the stated conditions, we showed in \cite{luong2022structured} that the ROC curve is
	\begin{equation} \label{eq:ROC_detMF}
		\pd(\pfa | \rho, N) = Q_1 \mleft( \rho\sqrt{2N}, \sqrt{-2 \ln \smash{\pfa}} \mright),
	\end{equation}
	which depends on $N$ and $\rho$ only through $N\rho^2$.
\end{proof}

The emergence of $N\rho^2$ is less trivial than it appears, because from a fundamental point of view, there is no reason to expect that the particular combination $N\rho^2$ should play any role at all. For example, the probability density function of $\hat{\rho}$ is
\begin{multline} \label{eq:PDF_rho}
	f_{\hat{\rho}}(x | \rho, N) = 2 (N-1)(1-\rho^2)^N x(1-x^2)^{N-2} \\
		\times {}_2F_1(N, N; 1; \rho^2 x^2)
\end{multline}
where ${}_2F_1$ is the Gaussian hypergeometric function \cite{luong2022structured}. The combination $N\rho^2$ does not appear in this expression, nor is it obvious how $N\rho^2$ would arise as $N \to \infty$ and $\rho \to 0$.

These results suggest the following tradeoff between $\rho$ and $N$: if $\rho$ is increased by a factor of $\alpha$, then the number of integrated samples required to achieve the same detection performance is reduced to $N/\alpha^2$. This explains the experimental observation in \cite{luong2019roc}, where it was found that the QTMS radar prototype required an integration time about eight times shorter than that of a standard noise radar with the same transmit power. This improvement is a result of the fact that the QTMS radar improved $\rho$ by a factor of approximately three. (The expected reduction in the integration time should be a factor of nine, but the QTMS radar experiment had not accumulated enough experimental data to establish this with certainty.) The fact that detection performance depends on $N\rho^2$ is really a mathematical representation of the tradeoff between the cost of a radar and its detection performance. We showed in \cite{luong2022performance} that $\rho$ depends on the fidelity of the reference signal to the transmit signal, as well as the parameters in the radar range equation (transmit power, antenna gain, etc.). Hence, increasing $\rho$ really means building a better radar---which costs money. The payoff comes in reducing the required $N$, which means the radar can detect a given target more quickly. This would be helpful in detecting small, fast-moving targets that remain in the radar's field of view for small periods of time, such as unmanned aerial vehicles.

It is true that the preceding propositions require $N$ to be large and $\rho$ to be small. Although they appear to be only mathematical simplifications at first glance, these assumptions also have a physical significance. As shown in \cite{luong2022performance}, $\rho$ depends on the radar cross section and the range of the target. Hence, targets that are difficult to detect---because they are far away, small, or have low reflectivity---correspond to low values of $\rho$. But this is exactly the circumstance under which radars are most useful, and an understanding of their detection performance most critical. And it stands to reason that difficult-to-detect targets require longer integration times, which fits in with the requirement that $N$ be large. Of course, the case of small $N$ and large $\rho$ is also of importance, and we will explore this case as well. There are, in theory, two other cases: $\rho$ and $N$ either both large or both small. In the former case, the target will almost certainly be detected; in the latter case, one can hardly expect to see anything. These cases can safely be ignored.

\section{Probability of Detection vs.\ $N\rho^2$}
\label{sec:pd}

In the spirit of plots such as \cite[Fig.\ 2.6]{skolnik1962introduction} that show $\pd$ as a function of SNR, we now present plots of $\pd$ as a function of $N\rho^2$. Naturally, SNR is a very different quantity from $N\rho^2$, but both quantities help radar designers determine the parameters needed to achieve a given level of detection performance.

\begin{figure}[t]
	\centerline{\includegraphics[width=\columnwidth]{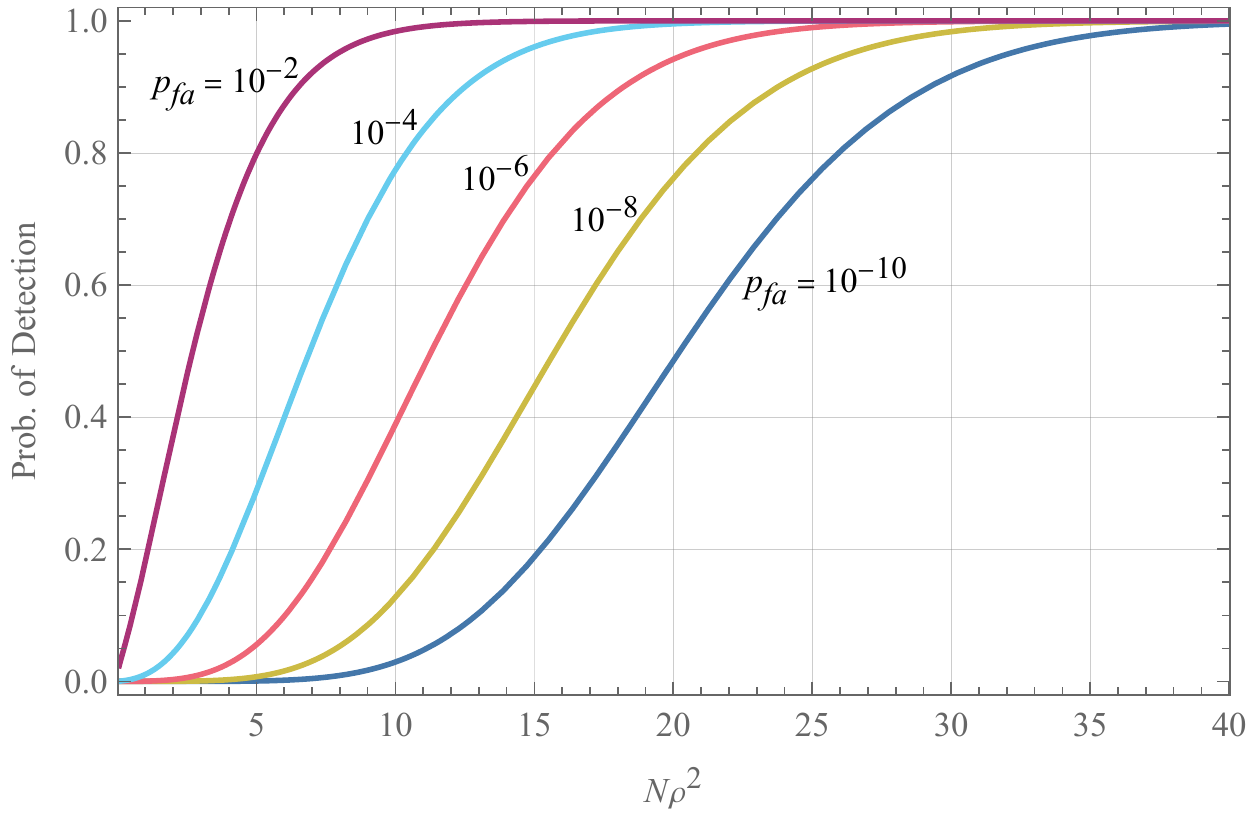}}
	\caption{Probability of detection as a function of $N\rho^2$ for the detector $D_0$. From left to right, the curves are for false alarm probabilities of $10^{-2}$, $10^{-4}$, $10^{-6}$, $10^{-8}$, and $10^{-10}$.}
	\label{fig:pd_Nrho2_erf}
\end{figure}

\begin{figure}[t]
	\centerline{\includegraphics[width=\columnwidth]{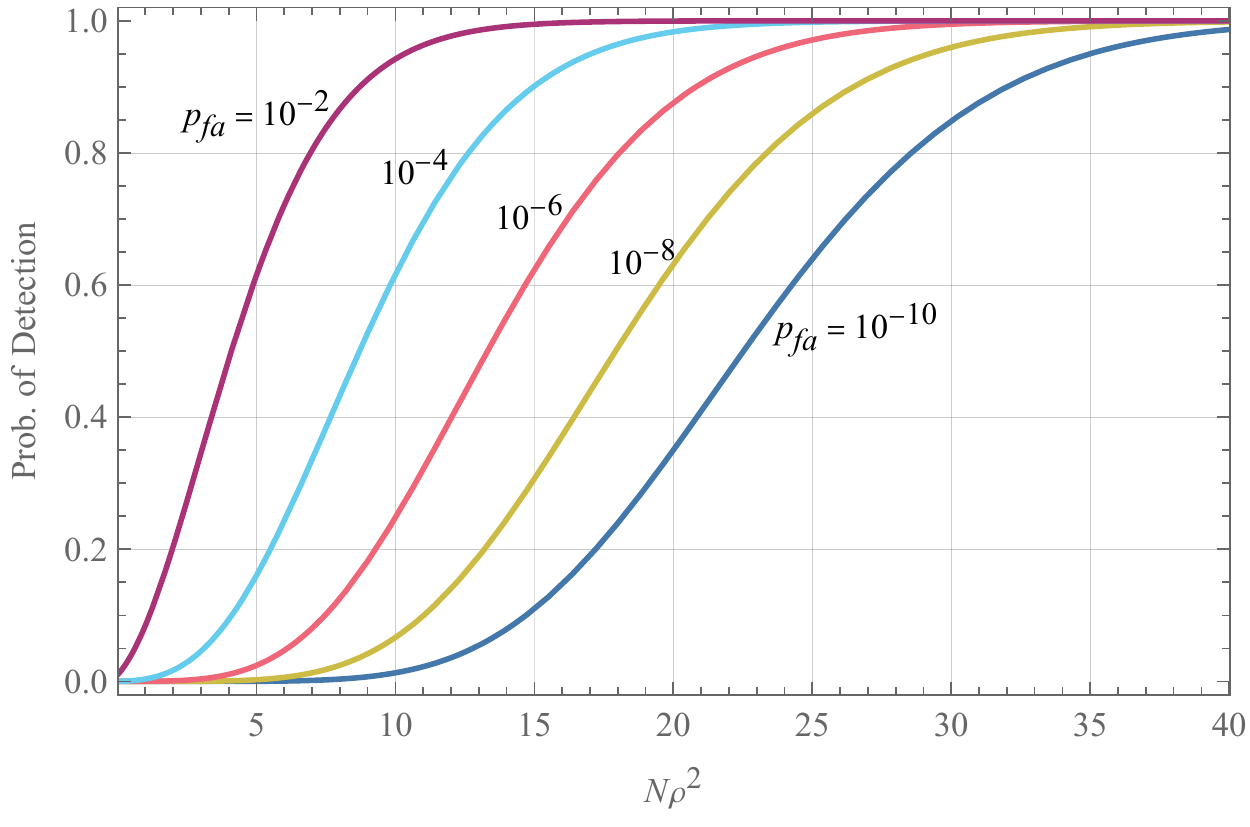}}
	\caption{Probability of detection as a function of $N\rho^2$ for the detectors $\hat{\rho}$ and $D_\mathrm{MF}$. From left to right, the curves are for false alarm probabilities of $10^{-2}$, $10^{-4}$, $10^{-6}$, $10^{-8}$, and $10^{-10}$.}
	\label{fig:pd_Nrho2_marcum}
\end{figure}

In Fig.\ \ref{fig:pd_Nrho2_erf}, we plot $\pd$ as a function of $N\rho^2$ using the ROC curve expression \eqref{eq:ROC_0_small}, which describes the detection performance of $D_0$ in the limit of large $N$ and small $\rho$. Similarly, Fig.\ \ref{fig:pd_Nrho2_marcum} shows plots of $\pd$ as a function of $N\rho^2$ for $\hat{\rho}$ using the expression \eqref{eq:ROC_detRho_small}. Since \eqref{eq:ROC_detRho_small} is equal to \eqref{eq:ROC_detMF}, Fig.\ \ref{fig:pd_Nrho2_marcum} is valid for $D_\mathrm{MF}$ as well. 

These plots can be used to obtain a quick estimate of the parameter values needed to achieve a given level of detection performance. As an example, we may read from Fig.\ \ref{fig:pd_Nrho2_marcum} that, in order to achieve $\pd = 0.95$ and $\pfa = 10^{-6}$ using $D_\mathrm{MF}$, we require at least $N\rho^2 \approx 25$. If we had a noise radar that achieves $\rho = 0.1$, we would need to integrate 2500 samples to obtain the specified $\pd$ and $\pfa$. This value of $N$ corresponds to an integration time of 2.5~ms if the radar receiver uses a sampling rate of 1~MHz.

\subsection{Approximating $\pd$ Using a Generalized Logistic Function}

In the ROC curve expressions \eqref{eq:ROC_0_small}, \eqref{eq:ROC_detRho_small}, and \eqref{eq:ROC_detMF}, the complementary error function and the Marcum $Q$-function appear. Neither of them can be expressed in terms of elementary functions, which makes them somewhat clumsy to work with. In particular, they cannot readily be inverted to determine $N\rho^2$ as a function of $\pd$ and $\pfa$.

It is true that both functions, particularly the error function, are well-studied in the literature. There are any number of excellent approximations available for these functions. However, these approximations must be chosen with care. For example,
\begin{equation}
	\erf(x) \approx \tanh \mleft( \frac{2x}{\sqrt{\pi}} \mright)
\end{equation}
is an good approximation for most values of $x$ and is moreover easily invertible. (The factor of $2/\!\sqrt{\pi}$ ensures that both functions agree to first order in their Taylor series.) But surprisingly, when substituted into \eqref{eq:ROC_0_small}, it leads to totally inaccurate estimates of the curves in Fig.\ \ref{fig:pd_Nrho2_erf}. This suggests that we should consider a slightly more sophisticated approach when approximating the error function and the Marcum $Q$-function.

In this subsection, we take inspiration from the fact that the ROC curves \eqref{eq:ROC_0_small}, \eqref{eq:ROC_detRho_small}, and \eqref{eq:ROC_detMF} are sigmoid functions, and consider approximations of the form
\begin{equation} \label{eq:gen_logistic}
	y(x | A, S, k, d) = \frac{A}{(1 + S e^{-kx})^d}.
\end{equation}
This is one of the generalizations of the logistic function studied in \cite{tjoerve2010unified}; the standard logistic function may be recovered by setting $A = S = k = d = 1$. Note that $y(x | A, S, k, d)$ is easily inverted:
\begin{equation}
	x = \frac{1}{k} \ln \mleft[ \frac{S}{(A/y)^{1/d} - 1} \mright].
\end{equation}

Our task is to find appropriate values of the parameters $A$, $S$, $k$, and $d$ such that $y(x | A, S, k, d)$ are good approximations for the expressions \eqref{eq:ROC_0_small}, \eqref{eq:ROC_detRho_small}, and \eqref{eq:ROC_detMF}. We may immediately choose $A = 1$ because this parameter controls the upper asymptote of the generalized logistic function, which we know must equal 1 because $\pd \leq 1$. In order to choose the other three parameters, the approach we will take is to numerically minimize the integral square errors
\begin{align}
	\epsilon_0 &= \frac{1}{L} \int_{0}^{L} \mleft[ \erfc \! \big[ \! \erfc^{-1}(2 p_\mathit{fa}) - \sqrt{x} \big] - y(x | S, k, d) \mright]^2 \, dx \\
	\epsilon_1 &= \frac{1}{L} \int_{0}^{L} \mleft[ Q_1 \big( \sqrt{2x}, \sqrt{-2 \ln \smash{\pfa}} \big) - y(x | S, k, d) \mright]^2 \, dx
\end{align}
where $x = N\rho^2$. In principle, the integration should run from zero to infinity. However, to avoid numerical issues associated with improper integrals, we take
\begin{equation}
	L = 15 - 3 \log_{10}(\pfa).
\end{equation}
This was chosen empirically (and somewhat arbitrarily) so that transition in the ROC curves from 0 to 1 is fully captured.

In the course of minimizing $\epsilon_0$ and $\epsilon_1$, we found that in many cases, $d \approx 2$. Therefore, we fix $d = 2$ and minimize only over $S$ and $k$. The resulting parameter values are summarized in Tables \ref{table:logistic_params_erf} (for $D_0$) and \ref{table:logistic_params_marcum} ($\hat{\rho}$ and $D_\mathrm{MF}$). The tables show that, for the given choices of $S$ and $k$, the error in the approximation is on the order of $10^{-4}$ or $10^{-5}$. Figs.\ \ref{fig:pd_Nrho2_logistic_erf} and \ref{fig:pd_Nrho2_logistic_marcum} show that the generalized logistic approximation results in excellent fits for the graphs of $\pd$ vs.\ $N\rho^2$.

\begin{table}[t]
	\centering
	\caption{Parameters used to fit the generalized logistic function to the ROC curve of $D_0$}
	\label{table:logistic_params_erf}
	\begin{tabular}{S[table-parse-only]S[table-format=2.6]S[table-format=1.6]S[table-format=2.5]}
		\toprule
		{$\pfa$} & {$S$} & {$k$} & {$\epsilon_0 \, (\times 10^5)$} \\
		\midrule
		e-1 & 0.979525 & 0.943077 & 17.2152 \\ 
		e-2 & 2.29804 & 0.603857 & 19.0271 \\ 
		e-3 & 4.01166 & 0.464131 & 13.4921 \\ 
		e-4 & 6.16081 & 0.385126 & 8.98033 \\ 
		e-5 & 8.96457 & 0.335449 & 6.20789 \\ 
		e-6 & 12.7266 & 0.301694 & 4.57958 \\ 
		e-7 & 17.3843 & 0.275690 & 3.67197 \\ 
		e-8 & 23.1463 & 0.255142 & 3.16590 \\ 
		e-9 & 25.6120 & 0.229508 & 5.98379 \\ 
		e-10 & 38.0617 & 0.223508 & 2.83859 \\
		\bottomrule
	\end{tabular}
\end{table}

\begin{table}[t]
	\centering
	\caption{Parameters used to fit the generalized logistic function to the ROC curves of $\hat{\rho}$ and $D_\mathrm{MF}$}
	\label{table:logistic_params_marcum}
	\begin{tabular}{S[table-parse-only]S[table-format=2.5]S[table-format=1.6]S[table-format=2.5]}
		\toprule
		{$\pfa$} & {$S$} & {$k$} & {$\epsilon_1 \, (\times 10^5)$} \\
		\midrule
		e-1 & 1.38363 & 0.642130 & 13.4719 \\ 
		e-2 & 3.13574 & 0.480237 & 15.2444 \\ 
		e-3 & 5.32653 & 0.392948 & 10.5543 \\ 
		e-4 & 8.19316 & 0.339720 & 7.14549 \\ 
		e-5 & 11.7823 & 0.302533 & 5.15448 \\ 
		e-6 & 16.2928 & 0.275043 & 4.03962 \\ 
		e-7 & 22.2072 & 0.254621 & 3.43057 \\ 
		e-8 & 28.7084 & 0.236497 & 3.11755 \\ 
		e-9 & 36.0844 & 0.221113 & 3.17315 \\ 
		e-10 & 52.9000 & 0.214206 & 4.98972 \\
		\bottomrule
	\end{tabular}
\end{table}

\begin{figure}[t]
	\centerline{\includegraphics[width=\columnwidth]{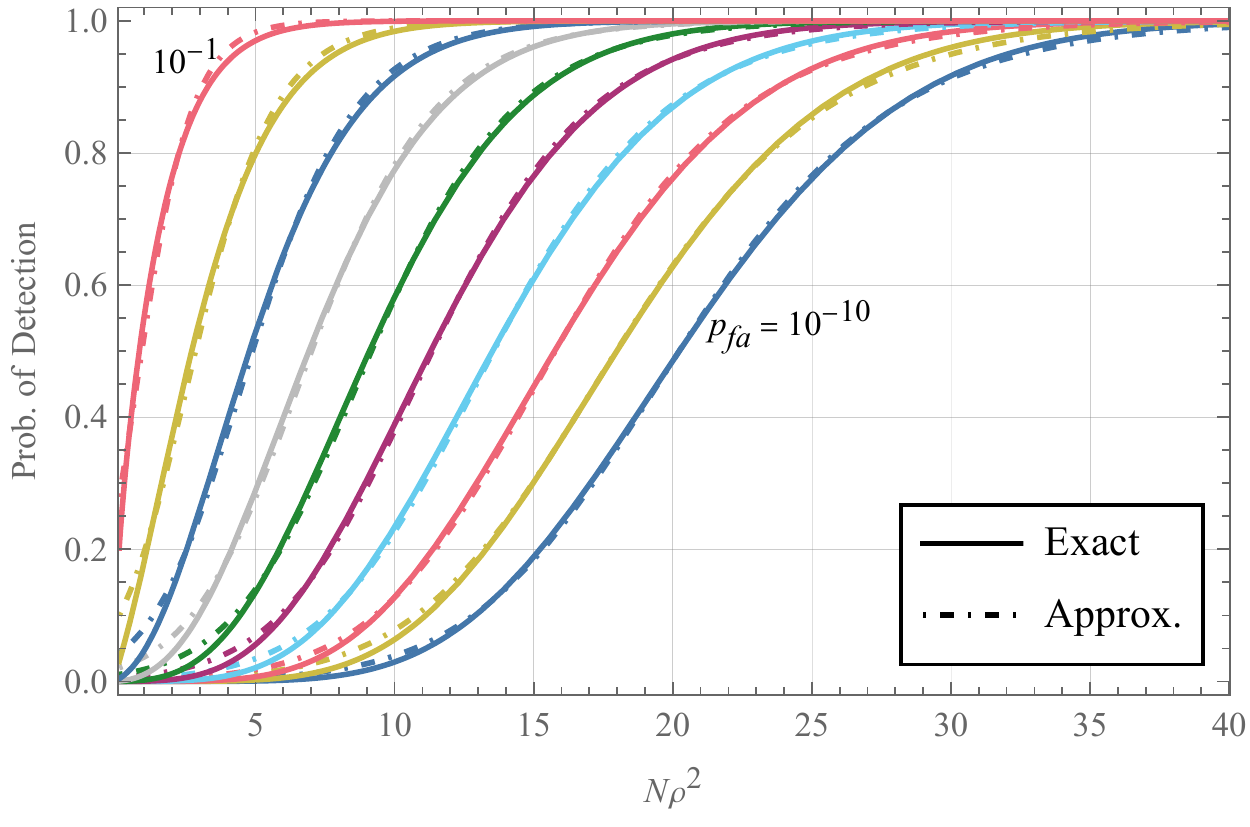}}
	\caption{Probability of detection as a function of $N\rho^2$ for the detector $D_0$, plotted together with approximations using the generalized logistic function. From left to right, the curves are for $\pfa = 10^{-1}, 10^{-2}, \dots, 10^{-10}$.}
	\label{fig:pd_Nrho2_logistic_erf}
\end{figure}

\begin{figure}[t]
	\centerline{\includegraphics[width=\columnwidth]{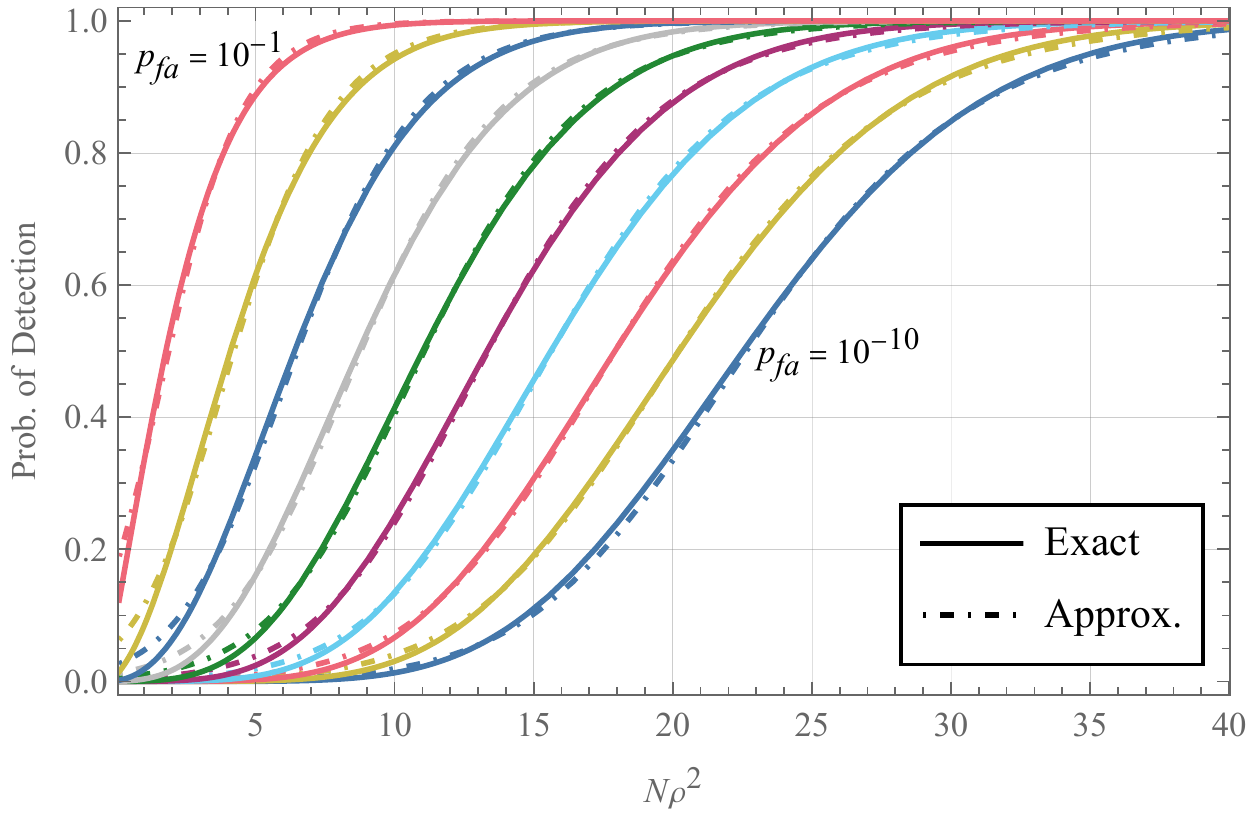}}
	\caption{Probability of detection as a function of $N\rho^2$ for the detectors $\hat{\rho}$ and $D_\mathrm{MF}$, plotted together with approximations using the generalized logistic function. From left to right, the curves are for $\pfa = 10^{-1}, 10^{-2}, \dots, 10^{-10}$.}
	\label{fig:pd_Nrho2_logistic_marcum}
\end{figure}

\subsection{Large $\rho$}

Recall that \eqref{eq:ROC_0_small}, \eqref{eq:ROC_detRho_small}, and \eqref{eq:ROC_detMF} were all derived under the assumption that $\rho$ is small and $N$ is large. When these assumptions are violated, we may expect deviations from Figs.\ \ref{fig:pd_Nrho2_erf} and \ref{fig:pd_Nrho2_marcum}. We found that for $\rho$ as large as 0.2, the deviation is minimal. However, when $\rho$ is increased further, the deviation becomes more evident. 

\begin{figure}[t]
	\centerline{\includegraphics[width=\columnwidth]{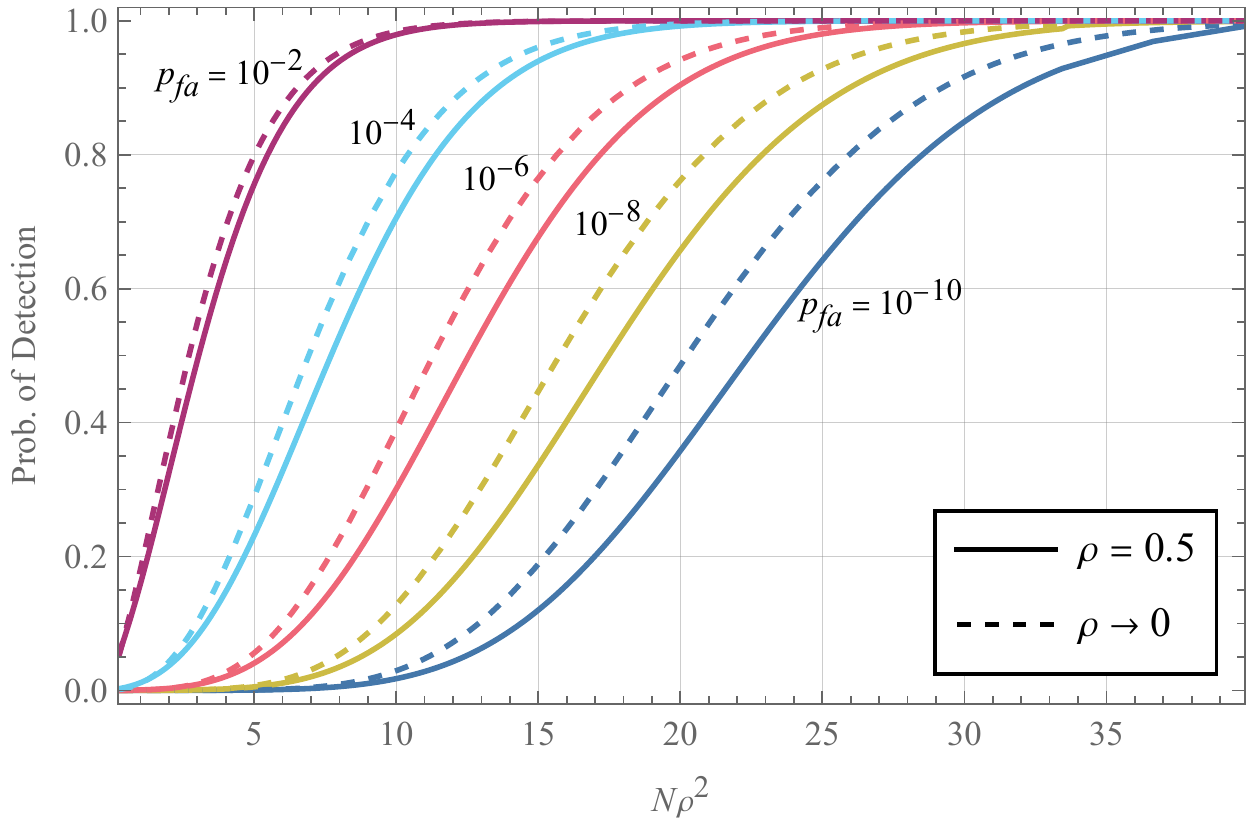}}
	\caption{Probability of detection as a function of $N\rho^2$ for $D_0$ when $\rho = 0.5$ (solid lines) and $\rho \to 0$ (dashed lines). From left to right, the curves are for false alarm probabilities of $10^{-10}$, $10^{-8}$, $10^{-6}$, $10^{-4}$, and $10^{-2}$.}
	\label{fig:pd_large_rho_detNP}
\end{figure}

\begin{figure}[t]
	\centerline{\includegraphics[width=\columnwidth]{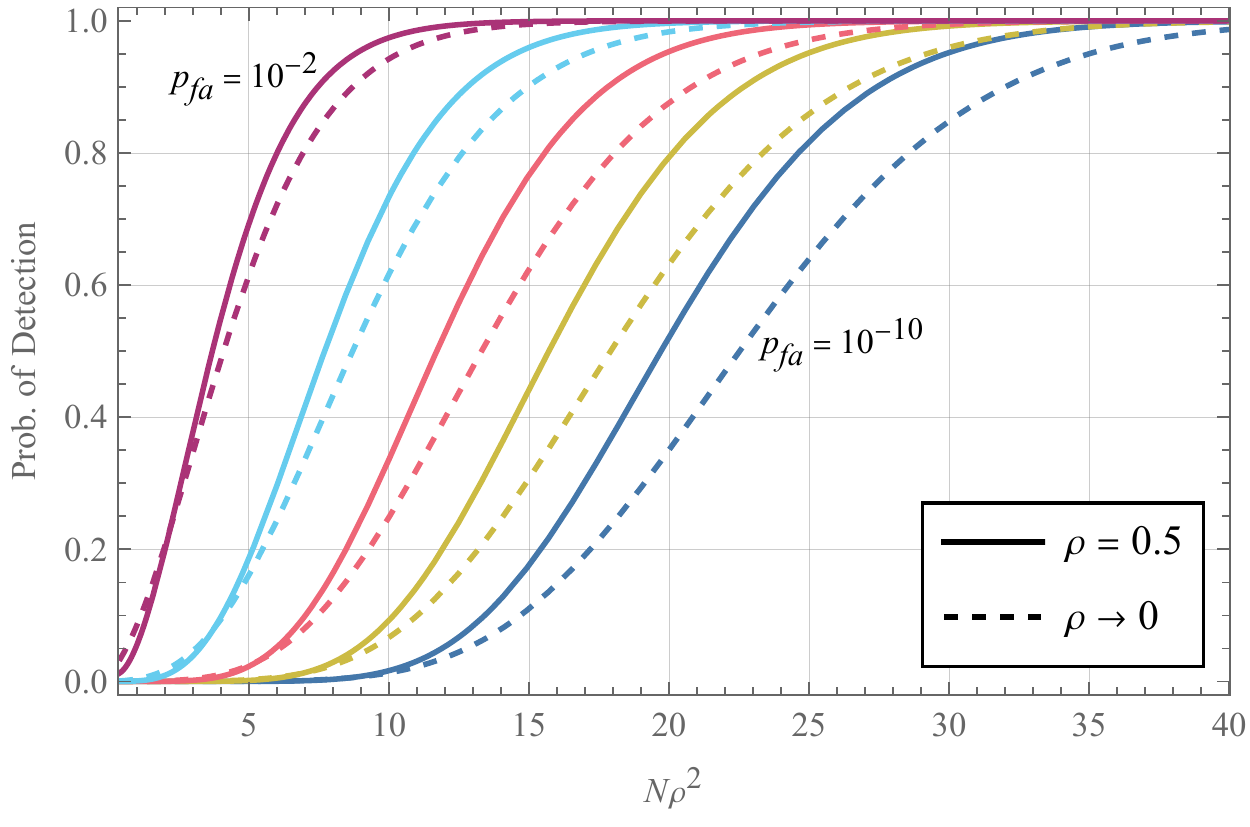}}
	\caption{Probability of detection as a function of $N\rho^2$ for the detector $\hat{\rho}$ when $\rho = 0.5$ (solid lines) and $\rho \to 0$ (dashed lines). From left to right, the curves are for false alarm probabilities of $10^{-10}$, $10^{-8}$, $10^{-6}$, $10^{-4}$, and $10^{-2}$.}
	\label{fig:pd_large_rho_detRho}
\end{figure}

\begin{figure}[t]
	\centerline{\includegraphics[width=\columnwidth]{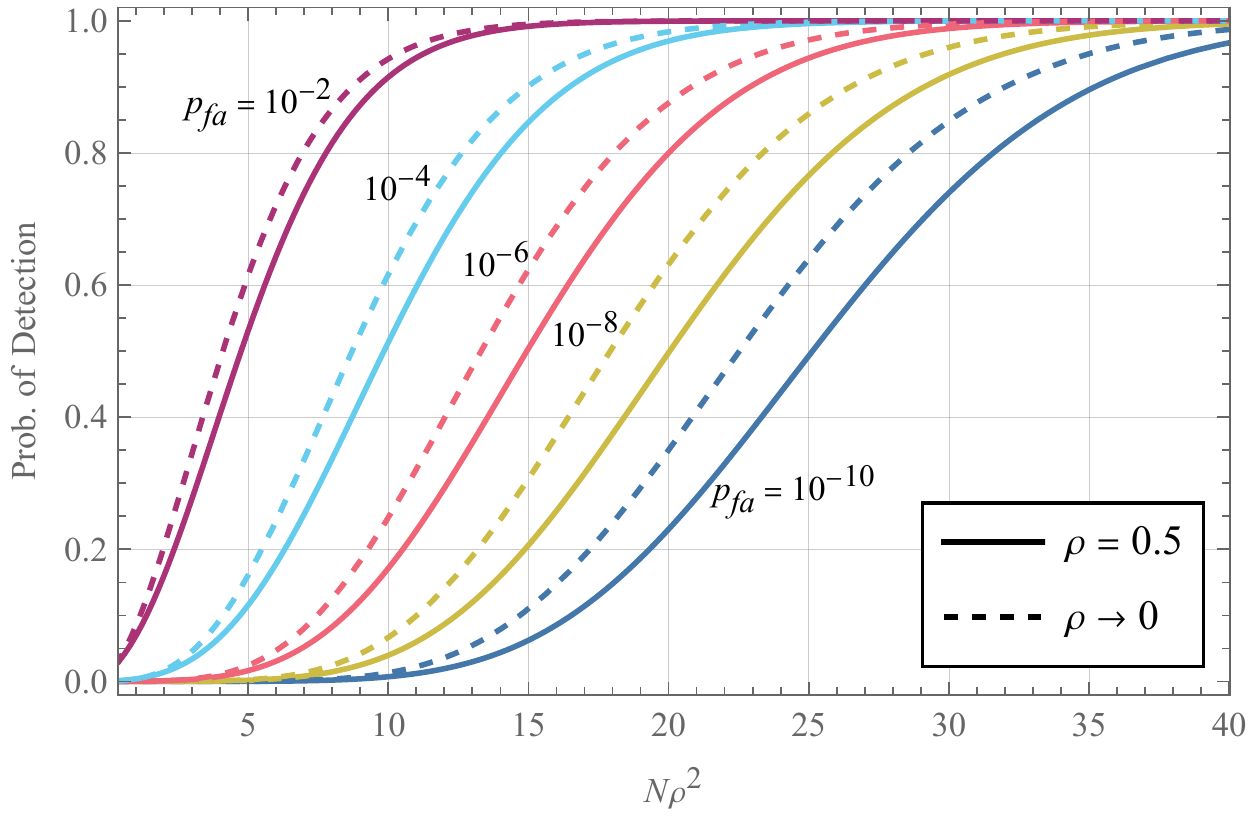}}
	\caption{Probability of detection as a function of $N\rho^2$ for $D_\mathrm{MF}$ when $\rho = 0.5$ (solid lines) and $\rho \to 0$ (dashed lines). From left to right, the curves are for false alarm probabilities of $10^{-10}$, $10^{-8}$, $10^{-6}$, $10^{-4}$, and $10^{-2}$.}
	\label{fig:pd_large_rho_detMF}
\end{figure}

In Figs.\ \ref{fig:pd_large_rho_detNP}, \ref{fig:pd_large_rho_detRho}, and \ref{fig:pd_large_rho_detMF}, we set $\rho = 0.5$ and plot $\pd$ vs.\ $N\rho^2$ for $D_0$, $\hat{\rho}$, and $D_\mathrm{MF}$ respectively. Since Props.\ \ref{prop:rhoN2_detNP}--\ref{prop:rhoN2_detMF} apply only when $\rho$ is small, the relevant ROC curves depend on $\rho$ and $N$ separately. Therefore, these plots are really of $\pd$ vs.\ $N$, but with the horizontal axis rescaled by a factor of $\rho^2$. (Note that the plots were obtained by numerically integrating the exact probability distributions for each detector; no approximations were used.) We see that when $\rho$ is large, the small-$\rho$ plots underestimate the required $N\rho^2$ for $D_0$ and $D_\mathrm{MF}$, but \emph{over}estimate it for the detector $\hat{\rho}$. This may be explained on the grounds that $D_0$ is the locally most powerful test statistic for $\rho$ close to zero, so there is no guarantee that it will perform well when $\rho$ is large. Similarly, the matched filter output $D_\mathrm{MF}$ maximizes SNR, not probability of detection, so one need not be surprised that it underperforms the maximum likelihood ratio detector $\hat{\rho}$.

\section{Conclusion}
\label{sec:conclusion}

The main result of this paper is that, when a noise-type radar integrates a large number of samples $N$ and its correlation coefficient $\rho$ is small, the detection performance of the radar depends on $N$ and $\rho$ only through the combination $N\rho^2$. This holds for a variety of detectors, including the three detectors listed in Sec.\ \ref{sec:detectors}. Since the ROC curves for these detectors depend only on $N\rho^2$, there is a simple tradeoff between $N$ and $\rho$: an improvement in $\rho$ by a factor of $\alpha$ is the same as improving $N$ by a factor of $\alpha^2$.

Having established the importance of the quantity $N\rho^2$, we show that we can plot the probability of detection as a function of $N\rho^2$. This is similar to the well-known plots of $\pd$ vs.\ SNR that are common in the radar literature. Moreover, we give simple approximations of $\pd$ based on a generalization of the logistic function. Finally, we showed how our results would change when the small-$\rho$, large-$N$ assumption is violated. These results can serve as aids for designing noise-type radars that are required to achieve given values of $\pd$ and $\pfa$.

The appearance of the specific combination $N\rho^2$ is not obvious from fundamental considerations such as the probability distributions for $\hat{\rho}$ or the other detectors. It is for this reason that the results in this paper are nontrivial. However, the fact that $N\rho^2$ emerges in the small-$\rho$, large-$N$ limit for a variety of detectors (including the ones in Sec.\ \ref{sec:detectors}) suggests that there may be some underlying theoretical explanation, independent of the specific detector being used. The search for such an explanation is a promising subject of future work, as it would probably throw an interesting light on the mathematical foundations of noise-like radars.

\section*{Acknowledgment}

This work was supported by the Natural Science and Engineering Research Council of Canada (NSERC). D.\ Luong also acknowledges the support of a Vanier Canada Graduate Scholarship.

\bibliographystyle{ieeetran}
\bibliography{qradar_refs,own_refs}
	
\end{document}